\documentclass[11pt]{article}
\usepackage{amsmath,amssymb,amsthm}
\usepackage{graphicx}
\usepackage{comment}

\theoremstyle{plain}
\newtheorem{theorem}{Theorem}
\newtheorem*{theorem*}{Theorem}
\newtheorem{lemma}{Lemma}
\newtheorem{corollary}{Corollary}
\newtheorem*{corollary*}{Corollary}
\newtheorem{definition}{Definition}

\theoremstyle{definition}
\newtheorem{example}{Example}
\newtheorem*{remark*}{Remark}

\def\R{{\mathbb R}}

\def\Q{{\mathbb Q}}
\def\l{{\lambda}}

\title{Linear recurrent cryptography:\\
golden-like cryptography for higher order linear recurrences}
\author{Sergiy Koshkin* and Daniel Rodriguez\textdagger\\
\\
*Corresponding author\\
Department of Mathematics and Statistics\\
University of Houston-Downtown\\
One Main Street\\
Houston, TX 77002\\
e-mail: koshkins@uhd.edu\\
\\
\textdagger Department of Mathematics and Statistics\\
University of Houston-Downtown\\
One Main Street\\
Houston, TX 77002\\
e-mail: rodriguezd153@gator.uhd.edu}
\date{}
\begin{document}

\maketitle
\begin{abstract}\
We develop matrix cryptography based on linear recurrent sequences of any order that allows securing encryption against brute force and chosen plaintext attacks. In particular, we solve the problem of generalizing error detection and correction algorithms of golden cryptography previously known only for recurrences of a special form. They are based on proving the checking relations (inequalities satisfied by the ciphertext) under the condition that the analog of the golden $Q$-matrix has the strong Perron-Frobenius property. These algorithms are proved to be especially efficient when the characteristic polynomial of the recurrence is a Pisot polynomial. Finally, we outline algorithms for generating recurrences that satisfy our conditions.
 
\bigskip

\end{abstract}

\section{Introduction}

Golden cryptography, as originally proposed by Stakhov \cite{St06,St07}, is a type of matrix encryption where the matrices are $2\times2$ and their entries are consecutive terms of the Fibonacci sequence. It has attractive error detection and correction properties, and was applied to creating digital signatures \cite{BAh}, and, with some modifications, to image encryption and scrambling \cite{Mish}. However, it is vulnerable to brute force and chosen plaintext attacks \cite{ReySan,Tah} due to a small number of parameters in the encryption key that are relatively easy to guess or back engineer. 

Some external countermeasures were proposed in  \cite{Moh,Sud}, and 
in \cite{KS17} using second order linear recurrences more general than the Fibonacci was proposed to increase the number of key parameters while preserving the error detection/correction properties. Some special higher order recurrences, such as $p$-Fibonacci \cite{BP09} and Tribonacci \cite{BD14}, have also been used. However, for security purposes it is desirable to develop matrix encryption based on more or less general linear recurrences if the attractive properties of golden cryptography can be preserved. 

In this paper, we accomplish this task. First, we prove that the {\it coding matrices} whose entries are consecutive terms of order $k$ linear recurrent integer sequences are exactly of the form $M_n=\begin{pmatrix}L^{n+k-1}x_0 & \dots & L^nx_0 \end{pmatrix}$, where $L$ is a matrix that naturally generalizes the $Q$-matrix of golden cryptography and $x_0$ is an initial vector (Theorem \ref{RecMat} and Corollary \ref{RtoL}). Simple conditions on $L$ and $x_0$ guarantee that $M_n$ are invertible, as needed for decryption, and have entries of feasible size. In the special case when $L$ is of companion form, $M_n$ are symmetric and their entries are terms of a single recurrent sequence, as in golden cryptography.

A major problem that we had to solve was to generalize the so-called {\it checking relations} of golden cryptography (inequalities satisfied by the ciphertext entries) upon which its error detection and correction algorithms are based. It turned out that for this $L$ must have the {\it strong Perron-Frobenius property} \cite{Nou06}, i.e. have a simple dominant positive eigenvalue with a strictly positive eigenvector. The $Q$-matrix, the unimodular $2\times2$ matrices that arise in \cite{KS17} and other generalizations, and the special higher order matrices of \cite{BD14,BP09} happen to have this property, but they are all very special cases of a general template. By the Perron-Frobenius theory \cite[ch.\,9]{Lanc}, \cite[1.4]{Minc}, primitive matrices with non-negative entries have the strong Perron-Frobenius property, and even matrices with some negative entries can have it as well \cite{Nou06}.

The explicit computations traditionally used to establish the checking relations are exceedingly cumbersome already for the special third order sequences \cite{BD14,BP09}, and become intractable for higher orders. We propose a different approach based on spectral theory, and our choice of $M_n$ ensures that the checking relations take a particularly simple form (Theorem \ref{kDCheckRel}). Moreover, further analysis of the error correction algorithm shows that the range where the checking relations locate the correct value shrinks to a point asymptotically only when $L$ satisfies an additional {\it Pisot condition}: all of its subdominant eigenvalues are confined to the interior of the unit disk (Theorem \ref{RangeShrink}). 

The characteristic polynomials of strong Perron-Frobenius matrices with this property are known as Pisot polynomials \cite{AK08}. They have been extensively studied in number theory and are well understood. In particular, there are explicit formulas and efficient algorithms for generating them \cite{Boyd96,HS21}. Our result also clarifies the role of the unimodularity condition on $L$ imposed in \cite{KS17}, it implies the Pisot condition for $2\times2$ matrices.

From our general perspective, the $Q$-matrix played three different roles in the generation of the coding matrices. It was the initial matrix $M_0$, the generating matrix $L$, and (the transpose of) another generating matrix $R$ defined by $M_n=L^nM_0=M_0R^n$. In general, all three matrices are different, $L$ is not necessarily symmetric, and $M_n$ are not necessarily powers of a single matrix. These novelties require suitable changes that we work out, we also propose a number of methods for generating $L$ and $x_0$ that satisfy all our conditions. 

The paper is organized as follows. In Section \ref{Gold} we review the main features of golden cryptography. It also serves as preliminaries for introducing some standard notation and terminology of cryptography and linear algebra. In Section \ref{RecSym} we consider a special case where $M_n$ are symmetric and $L$ is the companion matrix of a recurrence. In Section \ref{RecGen} this is generalized to coding matrices whose rows are segments of different recurrent sequences that satisfy the same recurrence. In Section \ref{CheckLin} we prove the main results of the paper on the checking relations for general coding matrices. The transition ratio, which generalizes the golden ratio, is introduced in Section \ref{TransRat} and its role in regulating the size of $M_n$ entries is highlighted. Section \ref{ErrCor} develops error detection and correction algorithms based on the checking relations and illustrates their application with examples. Generation of the coding matrices that satisfy the strong Perron-Frobenius and Pisot conditions is discussed in Section \ref{GenEnc}, and also illustrated by examples. Finally, in the last section we summarize the conclusions.

\section{Golden cryptography}\label{Gold}

In this section, we review the main points of golden cryptography \cite{St07} with an eye on the general setting. Recall that the Fibonacci sequence is defined by the linear recurrence $F_{n+2}=F_{n+1}+F_n$ with the initial values $F_0=0$ and $F_1=1$. For the purposes of encryption we assemble the terms of the sequence into the {\it golden matrices}:
\begin{equation}\label{FibMn}
M_n:=\begin{pmatrix}
F_{n+2}  &  F_{n+1}\\
F_{n+1}  &  F_n
\end{pmatrix}.
\end{equation}
The recurrence relation can be written in a first order form in terms of these matrices, namely $M_{n+1}=QM_n$, where $Q:=\begin{pmatrix}
1  &  1\\
1  &  0\end{pmatrix}$ is called the {\it $Q$-matrix}. By induction, we then have $M_n=Q^nM_0$. Note that $Q$ is symmetric, and we also have a happy coincidence:
$$
M_0=\begin{pmatrix}
F_{2}  &  F_{1}\\
F_{1}  &  F_0
\end{pmatrix}
=\begin{pmatrix}
1  &  1\\
1  &  0\end{pmatrix}=Q.
$$
As a result, $M_n=Q^{n+1}$. It is convenient to extend the sequence to the negative values of $n$ while preserving the recurrence $M_{n+1}=QM_n$, this leads to $M_{-n}=Q^{-n}M_0=Q^{-n+1}$.

For the purposes of encryption, the plaintext is digitized (e.g. by using ASCII codes) and split into blocks of four numbers each that are then arranged into plaintext matrices $P=\begin{pmatrix}p_{11}&p_{12}\\p_{21}&p_{22}\end{pmatrix}$. The encryption key is a (large) natural number $n$, and the corresponding ciphertext matrix is computed as $C=PM_n$. Since $M_n$ are all powers of $Q$, which is invertible, so are they, and the decryption is done simply as $P=CM_n^{-1}$. Moreover, 
\begin{equation}\label{FibMn-1}
M_n^{-1}=(Q^nM_0)^{-1}=M_0^{-1}Q^{-n}=M_0^{-1}M_{-n}M_0^{-1}=M_{-n}Q^{-2},
\end{equation}
i.e. the decryption can be performed by computing the Fibonacci sequence backwards and inverting $Q$. We note another happy coincidence that $M_n$ commute with $Q$ because they are its powers. Explicitly, we find 
\begin{equation}\label{FibMn-1Exp}
M_n^{-1} = \left(-1\right)^{n+1}\begin{pmatrix} F_n & -F_{n+1}\\-F_{n+1} & F_{n+2} \end{pmatrix}.
\end{equation}

For error detection and correction we need a deeper property of the coding matrices. It can be derived as follows. Since $P=CM_n^{-1}$ and the entries of $P$ are non-negative integers we get four inequalities for the entries of $C$ using \eqref{FibMn-1Exp}. They can be packaged into two double inequalities for their row ratios, namely 
\begin{equation}\label{FibIneq}
\frac{F_{n+1}}{F_n} \leq \frac{c_{11}}{c_{12}}\,, \frac{c_{21}}{c_{22}} \leq \frac{F_{n+2}}{F_{n+1}}
\end{equation}
for odd $n$. For even $n$ the inequality signs are reversed. These are the {\it checking relations in the inequality form}.
The recipient of the ciphertext can test both checking relations and determine in which row the single error, if any, is located. The two suspect entries can then be estimated by assuming the other one to be correct. 

A well-known property of the Fibonacci sequence, which goes back to Kepler \cite{FV11}, is that $\lim_{n\to\infty}\frac{F_{n+1}}{F_n}=\varphi$, where $\varphi=1.618\dots$ is the golden ratio and the convergence is exponential. Therefore, the bounds in \eqref{FibIneq} get tighter and tighter with increasing $n$, and ultimately shrink to a single value, the golden ratio. Thus, for large $n$ we get the {\it asymptotic form of checking relations}:
\begin{equation}\label{FibCR}
\frac{c_{11}}{c_{12}}\approx\frac{c_{21}}{c_{22}}\approx\varphi.
\end{equation}
\noindent The next examples briefly illustrates the main ideas.
In the traditional golden cryptography the determinant of $P$ is often transmitted as additional checking data to enhance error correction, but we do not discuss it here because its utility diminishes for higher order recurrences, see Section \ref{ErrCor}.

\section{Linear recurrent sequences and symmetric\\ coding matrices}\label{RecSym}

We will now formally extend the encryption scheme of golden cryptography to higher order linear recurrent sequences. In this section, we consider the special case of symmetric coding matrices. They are generated based on a single sequence, and the analogy with golden cryptography is most transparent. However, the methods we develop will carry over to the more general case studied in Section \ref{RecGen}.

Recall that a $k$-th order linear recurrent sequence satisfies
\begin{equation}\label{kLinRec}
X_{n+k} = a_{k-1}X_{n+k-1}+\dots+a_1X_{n+1}+a_0X_n
\end{equation} 
for some coefficients $a_i$ \cite[ch.\,6]{AB}. The {\it characteristic (companion) polynomial of the recurrence} is defined as
\begin{equation}\label{CharPolRec}
f(z):=z^k-a_{k-1}z^{k-1}-\dots-a_1z-a_0.
\end{equation}
It is convenient to collect the initial values into the {\it initial vector} $x_0=\begin{pmatrix}X_{k-1} & \dots & X_{1} & X_0\end{pmatrix}^T$, which, together with $a_i$, determines the entire sequence.
For the purposes of cryptography, we are mostly interested in integer (even positive integer) sequences, so it is natural to assume that $a_i$ and the entries of $x_0$ are integers, although formal considerations in this and the next section do not depend on this. 
\begin{definition}\label{SymCodMn}
Symmetric coding matrices $M_n$ associated with a linear recurrent sequence $X_n$ are defined as 
\begin{equation}\label{SymMn}
M_n:=\begin{pmatrix}
X_{n+2k-2}  & \dots   & X_{n+k-1}\\
\vdots  &  \ddots & \vdots\\
X_{n+k-1}   & \dots  &  X_n
\end{pmatrix}. 
\end{equation} 
$M_0$ will be called the initial matrix.
\end{definition}
\noindent The definition is in direct analogy to the golden matrices \eqref{FibMn}. Note that the columns of $M_n$ are segments of the same recurrent sequence shifted by a single entry from one column to the next. This means that each one is obtained from the one preceding by multiplying the latter by a matrix of special form. 
\begin{definition}\label{SymLtransf}
The {\it left companion matrix} of a linear recurrence \eqref{kLinRec} is defined as
\begin{equation}\label{Ltransf}
L:=\begin{pmatrix}
a_{k-1}  & a_{k-2} & \dots & a_1  &  a_0\\
1  & 0 & \dots & 0  &  0\\
0  & 1 & \dots & 0  &  0\\
\vdots & \vdots &  \ddots  & \vdots & \vdots\\
0  & 0 & \dots & 1  &  0
\end{pmatrix}
\end{equation} 
\end{definition}
\begin{remark*}
This is essentially the companion matrix of the characteristic polynomial \eqref{CharPolRec} of the recurrence, see \cite[7.1]{HK}. Under the more common convention, the companion matrix of a polynomial has the coefficients $a_i$ appear in the last row, and the $1$-s appear above the diagonal rather than below. Our form is the result of following Stakhov's convention to place the sequence terms into vectors in descending order. Accordingly, one can recover the usual companion form by reversing the order of the indices. 
\end{remark*}
\noindent Note that for $k>2$ the matrix $L$ is non-symmetric in non-degenerate cases, and hence cannot be equal to the initial matrix $M_0$, which is. The $Q$-matrix of golden cryptography, therefore, splits into two different ones, $L$ and $M_0$. With the above notation, we can write 
\begin{equation}\label{MnCyc}
M_n:=\begin{pmatrix}L^{n+k-1}x_0 & \dots & L^nx_0\end{pmatrix},
\end{equation}
hence the entire coding sequence is determined by $L$ and $x_0$. In particular, 
\begin{equation}\label{M0Cyc}
M_0=\begin{pmatrix}L^{k-1}x_0 & \dots & Lx_0 & x_0\end{pmatrix}.
\end{equation} 
As in the golden case, $M_n$ satisfy a first order matrix recurrence $M_{n+1}=LM_n$, and $M_n=L^nM_0$.
\begin{remark*}
The authors of \cite{BD14,BP09} take the bare powers $L^n$ rather than $M_n$ as the coding matrices. This choice generalizes a different feature of the golden matrices, that they are powers of the $Q$-matrix. However, then the entries of $M_n$ have more complex expressions in terms of the sequence $X_n$, and removing the dependence on $M_0$ takes away additional parameters that make encryption more secure. 
\end{remark*}
\begin{example}\label{Bonacci}
Two higher order examples considered in the cryptographic literature are $k$-bonacci \cite{BD14} and $p$-Fibonacci \cite{BP09} sequences ($p:=k-1$). In the former all $a_i=1$, i.e. $X_{n+k} = X_{n+k-1}+\dots+X_{n+1}+X_n$, and in the latter all $a_i=0$ except $a_0=a_{k-1}=1$, i.e. $X_{n+k} = X_{n+k-1}+X_n$. In both examples $x_0:=(1,0\dots,0)^T$. Only the cases $k=3$, Tribonacci and $2$-Fibonacci, have been studied in any detail. For the Tribonacci the left transition and initial matrices are, respectively,
\begin{equation*}
L=\begin{pmatrix}
1  & 1 & 1\\
1  & 0 & 0\\
0  & 1 & 0
\end{pmatrix}, \ \ \ \
M_0=\begin{pmatrix}
2  & 1 & 1\\
1  & 1 & 0\\
1  & 0 & 0
\end{pmatrix};
\end{equation*}
and for the $2$-Fibonacci they are
\begin{equation*}
L=\begin{pmatrix}
1  & 0 & 1\\
1  & 0 & 0\\
0  & 1 & 0
\end{pmatrix}, \ \ \ \
M_0=\begin{pmatrix}
1  & 1 & 1\\
1  & 1 & 0\\
1  & 0 & 0
\end{pmatrix}.
\end{equation*}
Another natural recurrence to consider is $X_{n+k} = X_{n+1}+X_n$. For $k=3$ the matrices are:
\begin{equation*}
L=\begin{pmatrix}
0  & 1 & 1\\
1  & 0 & 0\\
0  & 1 & 0
\end{pmatrix}, \ \ \ \
M_0=\begin{pmatrix}
1  & 0 & 1\\
0  & 1 & 0\\
1  & 0 & 0
\end{pmatrix}.
\end{equation*}
This $L$, and its higher dimensional analogs, come up in the Perron-Frobenius theory due to a result of Wielandt \cite{NR19,Sch02}, so we will refer to the recurrence as the order $k$ Wielandt recurrence. Just as with $k$-bonacci and $p$-Fibonacci matrices the entries of its powers grow slowly, making them feasible for encryption. For $k=2$ all three reduce to the Fibonacci recurrence.
\end{example}
As in the golden case, we extend the sequence to the negative values so that $M_{n+1}=LM_n$ continues to hold. This extension exists uniquely if $L$ is invertible, which we must assume anyway. Recall that a vector $x$ is called {\it cyclic} for a $k\times k$ matrix $A$ if $x,Ax,\dots A^{k-1}x$ span the entire space \cite[7.1]{HK}.
\begin{lemma}[{\bf Invertibility of $M_n$}]\label{InvMn}
The coding matrices $M_n$ are invertible if and only if $L$ is invertible and $x_0$ is its cyclic vector. Moreover,
\begin{equation}\label{SymMn-1}
M_n^{-1}=M_0^{-1}M_{-n}M_0^{-1},
\end{equation}
\end{lemma}
\begin{proof}
Since $M_n=L^nM_0$ both $L$ and $M_0$ must be invertible. But $M_0=\begin{pmatrix}L^{k-1}x_0 & \dots & Lx_0 & x_0\end{pmatrix}$, so it is invertible if and only if $x_0$ is cyclic. The inversion formula is established by the same computation as in \eqref{FibMn-1}, with $Q$ replaced by $L$.
\end{proof}
\noindent Unlike in the golden case, the inversion formula identity cannot be simplified further because $L$ and $M_0$ do not necessarily commute. By the cofactor expansion in the last column, $\det L=(-1)^{k-1}a_0$, so it is invertible if and only if $a_0\neq0$. It is straightforward to check that $x_0:=(1,0\dots,0)^T$, which is typically chosen as the vector of initial values in golden cryptography and its generalizations \cite{BD14,BP09}, is always cyclic for a left companion matrix. But it has plenty more cyclic vectors, and their entries can provide welcome extra parameters for the encryption key.

Encryption and decryption proceed analogously to the golden cryptography. Numerical plaintext is split into blocks of $k^2$ numbers each and arranged into plaintext matrices $P$ (say, row by row, left to right). The encryption key is a triple: $k$ coefficients of the recurrence $a_i$, $k$ entries of the initial vector $x_0$, and a natural number $n$. In total, we have $2k+1$ parameters. The ciphertext matrix is computed as $C=PM_n$, and the decryption is performed by $P=CM_n^{-1}$. 
\begin{example}\label{2FibEncryp}
Consider the $p$-Fibonacci sequence with $p=2$. The recurrence is $X_{n+3} = X_{n+2}+X_n$, and for $x_0:=(1,0\dots,0)^T$ the sequence is 
$$
X_n=0,0,1,1,2,3,4,6,9,13,19,28,41,60,88,129,189,277,406\dots
$$
Forming the coding matrix for $n=15$ according to \eqref{GenMn} we find:
$$
M_{15} = \begin{pmatrix}
X_{19} & X_{18}  & X_{17}\\
X_{18} & X_{17}  & X_{16}\\
X_{17} & X_{16}  & X_{15}
\end{pmatrix}
=\begin{pmatrix} 406 & 277 & 189\\277 & 189 & 129\\189 & 129 & 88\end{pmatrix}\!.
$$

Suppose the word ALGORITHM is digitized using ASCII values $A = 65,\ L = 76,\ G = 71,\ O = 79,\ R = 82,\ I = 73,\ T = 84,\ H = 72,\ M = 77$, so the plaintext matrix is $P = \begin{pmatrix} 65 & 76 & 71\\79 & 82 & 73\\84 & 72 & 77 \end{pmatrix}$. Selecting $n = 15$ we compute 
$$
C=PM_{15} = \begin{pmatrix} 65 & 76 & 71\\79 & 82 & 73\\84 & 72 & 77 \end{pmatrix}\begin{pmatrix} 406 & 277 & 189\\277 & 189 & 129\\189 & 129 & 88 \end{pmatrix} = \begin{pmatrix} 60861 & 41528 & 28337\\68585 & 46798 & 31933\\68601 & 46809 & 31940 \end{pmatrix}
$$ 
To decrypt, we make use of \eqref{SymMn-1}. First, we extend the sequence to the negative indices, $X_{-n}=1,0,-1,1,1,-2,0,3,-2,-3,5,1,-8,4,9...$, and find $M_{-15}= \begin{pmatrix} 5 & 1 & -8\\1 & -8 & 4\\-8 & 4 & 9 \end{pmatrix}$. We know $M_0$ from Example \ref{Bonacci} and compute $M_0^{-1} = \begin{pmatrix} 0 & 0 & 1\\0 & 1 & -1\\1 & -1 & 0 \end{pmatrix}$. This is the only inversion that needs to be performed. Since the entries of $M_0$ are small compared to those of $M_n$ computational effort is saved. Moreover, according to \eqref{SymMn-1}, we now only need to perform two matrix multiplications to find $M_n^{-1}$. Thus, \begin{align*}
&M_{15}^{-1}=\begin{pmatrix} 9 & 5 & -12\\-5 & -7 & 21\\-12 & 21 & -5 \end{pmatrix}, \text{\ and}\\
&P=CM_{15}^{-1} = \begin{pmatrix} 60861 & 41528 & 28337\\68585 & 46798 & 31933\\68601 & 46809 & 31940 \end{pmatrix}\!\!\!\begin{pmatrix} 9 & 5 & -12\\5 & -7 & 21\\-12 & 21 & -5 \end{pmatrix} = \begin{pmatrix} 65 & 76 & 71\\79 & 82 & 73\\84 & 72 & 77 \end{pmatrix}\!\!.
\end{align*}
\end{example}

\section{General coding matrices}\label{RecGen}

The symmetric coding matrices of the previous section are generated by a single linear recurrent sequence. As a result, the left transition matrix $L$ had a special companion form, and its $k^2$ entries were completely determined by merely $k$ parameters $a_i$. We could instead pick an arbitrary matrix $L$ in addition to an arbitrary initial vector $x_0$, produce $M_0$ according to \eqref{M0Cyc}, and then generate the coding matrices $M_n:=L^nM_0$ as before. This increases the number of encryption key parameters from $2k+1$ to $k^2+k+1$. It turns out that such general coding matrices still have a nice characterization in terms of recurrent sequences. We discuss two alternative representations in this section.

\subsection{Row sequences}\label{RowSec}

While the symmetric coding matrices are generated by a single recurrent sequence, general ones involve $k$ different sequences, albeit all satisfying the same recurrence. For second order recurrences golden cryptography was extended to general coding matrices in \cite{KS17}.
\begin{theorem}[{\bf Recurrent coding matrices}]\label{RecMat} 
Let $L$ be a $k\times k$ matrix with the characteristic polynomial $\chi_L(z)=z^k-a_{k-1}z^{k-1}-\dots-a_1z-a_0$ and $x_0$ be a $k$-vector. Set 
\begin{equation}\label{M0Cyc2}
M_0:=\begin{pmatrix}L^{k-1}x_0 & \dots & Lx_0 & x_0\end{pmatrix}
\end{equation} 
and $M_{n+1}=LM_n$. Then 
\begin{equation}\label{GenMn}
M_n=L^nM_0=\begin{pmatrix}
X_{n+k-1}^{(k-1)} & \dots  & X_n^{(k-1)}\\[0.5em]
\vdots  & & \vdots\\
X_{n+k-1}^{(0)} & \dots  & X_n^{(0)}
\end{pmatrix},
\end{equation} 
where $X_n^{(i)}$ are some sequences satisfying the same recurrence 
$$
X_{n+k} = a_kX_{n+k-1}+\dots+a_1X_{n+1}+a_0X_n.
$$
\end{theorem}
\begin{proof}By induction, \begin{equation}\label{GenMnCol}
M_n:=\begin{pmatrix}L^{n+k-1}x_0 & \dots & L^nx_0 \end{pmatrix},
\end{equation}
so the columns of $M_n$ shift one place to the right in $M_{n+1}$. Therefore, $M_n$ have the form \eqref{GenMn} for some sequences $X_{n}^{(i)}$. By the Cayley–Hamilton theorem, 
$$
\chi_L(L)=L^k-a_{k-1}L^{k-1}-\dots-a_1L-a_0I=0.
$$
Applying both sides to $x_n:=L^nx_0$ and noting that $L^sx_n=x_{n+s}$, we see that the $i$-th entry of $\chi_L(L)x_n$ is 
$X_{n+k}^{(i)}-a_{k-1}X_{n+k-1}^{(i)}-\dots-a_1X_{n+1}^{(i)}-a_0X_{n}^{(i)}$, and all these entries are $0$. In other words, $X_{n}^{(i)}$ satisfy the recurrence for each $i$. 
\end{proof}
\noindent In the symmetric case, we picked a single $X_n$ and chose $X_0^{(i)}=X_i$, so the rows of $M_n$ were shifted segments of the same sequence. 
\begin{definition}\label{RowSeq}
Consider a matrix sequence $M_n$ of the form \eqref{GenMn}. We call the linear recurrent sequences $X_n^{(i)}$ its row sequences.
\end{definition}
\begin{example}\label{L1234}
Take $L = \begin{pmatrix}
1 &  2 \\ 3 & 4 \end{pmatrix}$. Its characteristic polynomial is $\chi_L(z) = z^2 -5z -2$, so $a_0 = 2$ and $a_1 = 5$. If we take the initial vector to be  $x_0 =  \begin{pmatrix}
1 & 0 \end{pmatrix}^T$ then $M_0 = \begin{pmatrix}
1&1\\3&0
\end{pmatrix}$ from \eqref{M0Cyc2}, which gives us the initial values for the recurrent row sequences. Computing them recursively, we find 
\begin{align*}
X_m^{(1)}&=1,1,7,37,199,1069,5743,30853,165751,890461,4783807...\\
X_m^{(0)}&=0,3,15,81,435,2337,12555,67449,362355,1946673,10458075... 
\end{align*}
For $n = 9$ this gives 
$$
M_{9} = \begin{pmatrix}
X_{10}^{(1)} & X_9^{(1)}\\[0.25em]
X_{10}^{(0)} & X_9^{(0)}\\
\end{pmatrix}
= \begin{pmatrix}
4783807 & 890461\\
10458075 & 1946673
\end{pmatrix}\!.
$$ 
The rapid growth of these sequences may cause problems in practice, when large values of $n$ need to be used, both for security purposes and for effective error correction. The underlying reason (see Section \ref{TransRat}) is that the spectral radius of our $L$ is quite large, $\approx5.372$. This has to be addressed when generating feasible $L$, see Section \ref{GenEnc}.
\end{example}

\subsection{Right companion matrices}\label{RiteComp}

We will now derive an alternative representation for $M_n$ that proves useful for both theoretical and practical purposes.
\begin{definition}\label{GenLtransf}
Consider a matrix sequence $M_n$. If $M_{n+1}=LM_n$ we call $L$ its left transition matrix, and if $M_{n+1}=M_nR$ we call $R$ its right transition matrix.
\end{definition}
\noindent Equivalently, $M_n=L^nM_0$ if $L$ is a left transition matrix, and $M_n=M_0R^n$ if $R$ is a right one. Of course, a general matrix sequence does not have either a left or a right transition matrix, and they may not be unique even when it does have them. However, if $M_0$ is invertible then $L$ and $R$, if any, are uniquely determined, and if there is one then there is also the other. Indeed, if $L$ is the left transition matrix then we set $R:=M_0^{-1}LM_0$, and compute by induction 
$$
L^nM_0=L^{n-1}LM_0=L^{n-1}M_0R=\dots=LM_0R^{n-1}=M_0R^n.
$$ 
For a symmetric coding sequence \eqref{SymMn} the left transition matrix coincides with the left companion matrix of the recurrence \eqref{Ltransf}, but in general it is not of the companion form, and is not determined by the recurrence alone. Indeed, it can be an arbitrary matrix with the characteristic polynomial coinciding with the characteristic polynomial of the recurrence. 

However, even for a general coding sequence \eqref{GenMn} the right transition matrix depends on the recurrence only. Indeed, that the row sequences satisfy the same recurrence can be expressed in the matrix form as  $M_{n+1}=M_nR$, or explicitly
\begin{equation*}
\begin{pmatrix}
X_{n+k}^{(k-1)} & \dots  & X_{n+1}^{(k-1)}\\[0.5em]
X_{n+k}^{(k-2)} & \dots  & X_{n+1}^{(k-2)}\\
\vdots  & & \vdots\\
X_{n+k}^{(0)} & \dots  & X_n^{(0)}\\
\end{pmatrix}\!\!=\!\!
\begin{pmatrix}
X_{n+k-1}^{(k-1)} & \dots  & X_n^{(k-1)}\\[0.5em]
X_{n+k-1}^{(k-2)} & \dots  & X_n^{(k-2)}\\
\vdots  & & \vdots\\
X_{n+k-1}^{(0)} & \dots  & X_n^{(0)}\\
\end{pmatrix}\!\!
\begin{pmatrix}
a_{k-1}  & 1 & 0 & \dots & 0\\
a_{k-2}  & 0 & 1 & \dots & 0\\
\vdots & \vdots &  \vdots  & \ddots & \vdots\\
a_1  & 0 & 0 & \dots & 1\\
a_0  & 0 & 0 & \dots & 0  
\end{pmatrix}\!\!.
\end{equation*} 
\begin{definition}\label{GenRtransf}
The {\it right companion matrix} of a linear recurrence \eqref{kLinRec} is defined as the transpose of its left companion matrix:
\begin{equation}\label{Rtransf}
R:=\begin{pmatrix}
a_{k-1}  & 1 & 0 & \dots & 0\\
a_{k-2}  & 0 & 1 & \dots & 0\\
\vdots & \vdots &  \vdots  & \ddots & \vdots\\
a_1  & 0 & 0 & \dots & 1\\
a_0  & 0 & 0 & \dots & 0  
\end{pmatrix}\!.
\end{equation}
\end{definition}
\noindent Thus, while the left transition matrix may be arbitrary, the right transition matrix is always of the right companion form. In the symmetric case $R=L^T$, and again coincides with the $Q$-matrix in the golden case. In general, all three matrices $M_0$, $L$ and $R$ are different. 

Recall that to enable decryption we need $L$ and $M_0$ to be invertible. The latter implies that $x_0$ must be a cyclic vector for $L$. Under this assumption, $R$ can be characterized as the unique matrix that satisfies $LM_0=M_0R$. Moreover, a converse of Theorem \ref{RecMat} holds.
\begin{corollary}\label{RtoL} Invertible matrices of the form \eqref{GenMn}, whose row sequences satisfy the same recurrence, are given by $M_n=L^nM_0$ for some matrix $L$ and $M_0$ of the form \eqref{M0Cyc2}.
\end{corollary}
\begin{proof}We know that $M_n=M_0R^n$ for $R$ in \eqref{Rtransf}. If we set $L:=M_0RM_0^{-1}$ then we have by induction
$M_n=L^nM_0$. Due to the companion structure of $R$, all columns after the first in $M_0R$ are the columns of $M_0$ shifted one place to the right. Since $LM_0=M_0R$ by construction it follows that the columns of $M_0$ can be obtained by successively multiplying its last column by $L$. Taking the last column as $x_0$, we thereby represented $M_0$ as in \eqref{M0Cyc2}.
\end{proof}
\noindent It follows that we can start with an arbitrary linear recurrence having $a_0\neq0$, and an arbitrary invertible matrix $M_0$, instead of $L$ and $x_0$, to get the same class of coding matrices \eqref{GenMn}. In particular, to generate the coding matrices $M_n$, we only need to generate $k$ recurrent sequences whose terms fill their rows. The entries of $R$ package the recurrence coefficients, and the rows of $M_0$ package the initial values for each.  

Moreover, if we take $M_0=R$ then $LM_0=M_0R$ implies that $L=M_0=R$, and all three matrices coincide. Then the coding matrices are bare powers $M_n=R^{n+1}$, as with the $Q$-matrix. This is similar to the setup adopted in \cite{BD14} for the Tribonacci recurrence, except the authors use the transpose $R^T$ in place of $R$ when taking the powers. This results in the checking relations having a more complicated form.
\begin{remark*} While left companion matrices \eqref{Ltransf} always have cyclic vectors this is not the case for general $L$. Nonetheless, matrices with cyclic vectors are generic. Recall that the minimal polynomial $\mu_L(z)$ of a matrix $L$ is the monic polynomial of the least degree such that $\mu_L(L)=0$. The matrix is called non-derogatory if $\mu_L(z)=\chi_L(z)$, or, equivalently, if its Jordan cells have distinct eigenvalues. It follows from linear algebra that a matrix has a cyclic vector if and only if it is non-derogatory \cite[7.1]{HK}. When this is the case non-cyclic vectors are confined to a subvariety of smaller dimensions, so almost every vector is cyclic. Moreover, if $\chi_L$ is irreducible over some subfield of $\R$ then $L$ is diagonalizable with distinct eigenvalues, and {\it every} non-zero vector with entries from that subfield is cyclic. This is particularly useful to us when the subfield is $\Q$.
\end{remark*}

\begin{example}\label{3x3M0R}
Using a random number generator (between $0$ and $2$) for the entries we generated $L = \begin{pmatrix}
2&1&2\\0&1&2\\2&2&2
\end{pmatrix}$. Its characteristic polynomial is computed to be $z^3 - 5z^2 + 4$, so $a_0 = -4$, $a_1 = 0$, and $a_2 = 5$. Since $a_0\neq0$ this $L$ is invertible. Moreover, the recurrence relation is $X_{n+3} = 5X_{n+2}-4X_n$, and the right companion matrix is $R = \begin{pmatrix}
5&1&0\\0&0&1\\-4&0&0\\
\end{pmatrix}$.

Note that $z^3 - 5z^2 + 4$ reduces over $\Q$ since $1$ is a root. Nonetheless, the standard initial vector $x_0 = \begin{pmatrix}
1\\0\\0
\end{pmatrix}$ is cyclic, and the initial matrix $M_0 = \begin{pmatrix}
8&2&1\\4&0&0\\8&2&0
\end{pmatrix}$ is invertible. Generating the row sequences by the recurrence with the initial values from $M_0$, we find, for example,
$
M_5 = \begin{pmatrix}
19292 & 3996 & 828 \\11300 & 2340 & 484 \\ 21632 & 4480 & 928\\
\end{pmatrix}.
$
\end{example}
Theorem \ref{RecMat} and Corollary \ref{RtoL} mean that $L,x_0,n$ and $M_0,a_i,n$ provide two alternative parametrizations of the encryption key, both with $k^2+k+1$ parameters. In the first representation, $L$ is arbitrary, but $M_0$ is of the special cyclic form \eqref{M0Cyc2}; in the second one, $M_0$ is arbitrary, but $R$ is of the special companion form \eqref{Rtransf}. For the second order recurrences these parametrizations and matrices were introduced in \cite{KS17}. 

\section{Checking relations}\label{CheckLin}

In this section, we turn from formal considerations to spectral conditions needed for the feasibility of encryption and error detection/correction, and prove the main technical results of the paper that generalize the checking relations of golden cryptography. Their inequality form turns out to hold universally when the coding and plaintext matrices have non-negative entries, and nothing needs to be assumed about $M_n$. But the asymptotic form requires much stronger conditions, such as the strong Perron-Frobenius and Pisot conditions for the left transition matrix $L$. Readers interested in applications to cryptography can skip most of this section, and only refer to definitions and the statements of Theorems \ref{kDCheckRel} and \ref{RangeShrink} as needed. 

\subsection{Checking relations: inequality form}

The usual method for deriving the two-sided inequalities \eqref{FibIneq} based on explicitly inverting the coding matrices quickly becomes too cumbersome in higher orders, as can be seen already for the simplest third order sequences in \cite{BD14} and \cite{BP09}. We give instead a very short proof based on an elementary inequality from \cite[2.1]{Minc} for $k$-tuples of real numbers $r_l,q_l$ with $q_l\geq0$:
\begin{equation}\label{MincIneq}
\min_{1\leq l\leq k}\left(\frac{r_l}{q_l}\right)
\leq \frac{\sum_{l=1}^kr_l}{\sum_{l=1}^kq_l}
\leq\max_{1\leq l\leq k}\left(\frac{r_l}{q_l}\right).
\end{equation} 
Here it is customary to assume $q_l>0$, but one can allow $q_l=0$ if $\frac{r_l}{0}$ is interpreted as $\pm\infty$ with the sign matching the sign of $r_l$, and $\frac00$ are simply dropped from the list under $\min$ and $\max$.
\begin{lemma}\label{CRatRange}
Suppose $C=PM$, where $P$ and $M$ have non-negative entries. Then 
\begin{equation}\label{minMmaxM}
\min_{1\leq l\leq k}\left(\frac{m_{lj}}{m_{lj'}}\right)
\leq \frac{c_{ij}}{c_{ij'}}
\leq\max_{1\leq l\leq k}\left(\frac{m_{lj}}{m_{lj'}}\right).
\end{equation} 
In other words, the ratios of same row entries of $C$ are between the minimal and the maximal ratios of same row entries of $M$ taken in the corresponding two columns.
\end{lemma}
\begin{proof}
By definition of matrix multiplication,
$$
\frac{c_{ij}}{c_{ij'}}
=\frac{\sum_{l=1}^kp_{il}\,m_{lj}}{\sum_{l=1}^kp_{il}\,m_{lj'}}\,.
$$
Applying \eqref{MincIneq} with $r_l=p_{il}\,m_{lj}$ and $q_l=p_{il}\,m_{lj'}$ we get the desired conclusion.
\end{proof}
\noindent If we apply Lemma \ref{CRatRange} to the golden matrices $M=M_n$ from \eqref{FibMn} with $j=1$, $j'=2$ then we immediately get the golden checking inequalities \eqref{FibIneq}. For the general coding matrices \eqref{GenMn} same row entries in $M_n$ are terms of the same row sequences $X_{n}^{(q)}$ with $m^{(n)}_{ij}=X_{n+k-j}^{(k-i)}$.

\subsection{Checking relations: asymptotic form}

Now let us turn to the asymptotic form of the checking relations \eqref{FibCR}. As with the Fibonacci recurrence, we need the ratios of consecutive terms in the row sequences to converge to a common limit that depends only on the recurrence. Recall that  
\begin{equation}\label{GenPMnCol}
M_n=\begin{pmatrix}L^{n+k-1}x_0 & \dots & L^nx_0 \end{pmatrix},
\end{equation}
i.e. consecutive columns of $M_n$ are obtained by applying consecutive powers of $L$ to the initial vector $x_0$. Therefore, the limit ratios between their corresponding entries are determined by
the asymptotic behavior of $L^nx_0$. This behavior can be conveniently studied using the decomposition of $x_0$ into generalized eigenvectors of $L$. 

For notation and terminology from linear algebra used in this and the next subsection see e.g. \cite{AB,Lanc}. Recall that a non-zero vector $x$ is called its eigenvector of a matrix $A$  when it satisfies $Ax = \lambda x$, and $\lambda$ is called the eigenvalue. All vectors annihilated by some power of $A-\l I$ are called generalized eigenvectors of $A$ corresponding to $\l$, and they form the generalized eigenspace of $A$. As a consequence of the canonical Jordan decomposition, any matrix has a basis of generalized eigenvectors.
\begin{definition}\label{domEigenValue}
An eigenvalue $\l$ is called simple when it has multiplicity $1$ in the characteristic polynomial. It is called dominant when $|\mu| < |\lambda|$ for all other eigenvalues $\mu$. 
\end{definition}
\noindent When an eigenvalue is simple the corresponding generalized eigenspace is one dimensional and is spanned by a single eigenvector. And if a matrix has real entries, the case of primary interest to us, then the dominant eigenvalue, if any, must be real. Otherwise, it cannot dominate its complex conjugate, which is also an eigenvalue. 
\begin{definition}[\cite{Nou06}]\label{StrPerFrob}
A matrix $L$ is said to have the strong Perron-Frobenius property if it has a simple positive dominant eigenvalue, and the corresponding eigenvector has strictly positive entries.
\end{definition}
\noindent The $Q$-matrix, and the left transition matrices of the $k$-Bonacci, $p$-Fibonacci, and any order Wielandt recurrences have the strong Perron-Frobenius property, as do all primitive matrices from the Perron-Frobenius theory. It is the condition we will use to derive the asymptotic form of the checking relations.
\begin{lemma}\label{RatLim}
Let $M_n$ be the general recurrent coding matrices \eqref{GenMn}. Assume that they are invertible and the left transition matrix $L$ of $M_n$ has the strong Perron-Frobenius property. Then for the ratios of terms in the row sequences of $M_n$ we have $\frac{X_{n+s}^{(i)}}{X_n^{(i)}}=\tau^s+o(n)$ with exponentially small $o(n)$, where $\tau$ is the dominant eigenvalue. 
\end{lemma}
\begin{proof} For simplicity, we sketch the proof for the case when all generalized eigenvectors are eigenvectors, only indicating changes needed for the general case. Let $\l_i$ be the eigenvalues of $L$ (possibly repeating) and $u_i$ be the corresponding eigenvectors, with $\l_1=\tau$ and $u_1=u$ being the dominant eigenvalue and the corresponding eigenvector. In the basis of eigenvectors $x_0=x_0^1\,u_1+\dots+x_0^k\,u_k$ and 
\begin{multline*}
L^nx_0=x_0^1\,\tau^nu+x_0^2\,\lambda_2^nu_2+\dots+x_0^k\,\lambda_k^nu_k\\
=\tau^n\left(x_0^1\,u + \left(\frac{\lambda_2}{\tau}\right)^n x_0^2\,u_2+\dots+ \left(\frac{\lambda_k}{\tau}\right)^nx_0^k\,u_k\right).
\end{multline*}
In general, the second and further terms may have  polynomials in $n$ multiplying the factors $\left(\frac{\lambda_i}{\tau}\right)^n$, but they still decrease exponentially since $\tau$ is dominant. 
Applying this asymptotic to the corresponding entries in $L^{n+s}x_0$ and $L^nx_0$ we find
$$
\frac{X_{n+s}^{(i)}}{X_n^{(i)}}=\frac{(L^{n+s}x_0)^i}{(L^{n}x_0)^i}=\frac{\tau^{n+s}\,x_0^1\,u^i+o(n)}{\tau^{n}\,x_0^1\,u^i+o(n)}=\tau^s+o(n),
$$
assuming that the dominant terms in the numerator and the denominator are non-zero. 

Since $u$ has strictly positive entries $u^i\neq0$, so it remains to show that $x_0^1\neq0$. By contradiction, suppose $x_0^1=0$. Then applying $L^n$ to the generalized eigenvector expansion of $x_0$ we see that $u$ component never appears in the expansions of $L^nx_0$ for any $n$. Therefore, $L^nx_0$ do not span the entire space and $x_0$ is not cyclic. But invertibility of $M_0$ implies that $x_0$ is cyclic by Corollary \ref{InvMn}, contradiction.
\end{proof}
\noindent Note that the proof implies that the entries of $M_n$ have the same sign for large $n$, namely the sign of $x_0^1$, and we can always make them positive by switching the sign of the initial vector if necessary. Combining Lemmas  \ref{CRatRange} and \ref{RatLim} we arrive at the announced main result.
\begin{theorem}\label{kDCheckRel}
Let the ciphertext matrix be $C=PM_n$, where $M_n$ are the general recurrent coding matrices \eqref{GenMn}, and $L$ be the left transition matrix of $M_n$. Assume the following:  

\noindent \textup{(i)} $P$ has non-negative entries; 

\noindent \textup{(i)} $M_n$ are invertible with positive entries; 

\noindent \textup{(iii)} $L$ has the strong Perron-Frobenius property.

Let $m_{ij}^{(n)}$ be the entries of $M_n$ and $\tau$ be the dominant eigenvalue of $L$. Then for the entries $c_{ij}$ of $C$ we have
\begin{equation}\label{minMmaxMn}
\min_{1\leq l\leq k}\left(\frac{m_{lj}^{(n)}}{m_{lj'}^{(n)}}\right)
\leq \frac{c_{ij}}{c_{ij'}}
\leq\max_{1\leq l\leq k}\left(\frac{m_{lj}^{(n)}}{m_{lj'}^{(n)}}\right),
\end{equation}
and both bounds converge to $\tau^{j'-j}$ exponentially fast for all $i,j,j'$ when $n\to\infty$. In particular, 
\begin{equation}\label{cijtauon}
\frac{c_{ij}}{c_{ij'}}=\tau^{j'-j}+o(n)
\end{equation} 
with exponentially small $o(n)$. 
\end{theorem}
\noindent We recover the more traditional checking relations for consecutive entries by taking $j'=j+1$.

\subsection{Pisot condition}\label{PisotCon}

One might hope that for large enough $n$ the ranges for $c_{ij}$ allowed by the checking relations will shrink, eventually leaving a single integer within them, and eliminate the need for trial and error in correction. However, this is not the case in general. True, the distance between the bounds in \eqref{minMmaxMn} does go to $0$ since they both converge to $\tau$. But if we are using $c_{ij'}$ to recover $c_{ij}$, the range of $c_{ij}$'s possible values gets scaled by $c_{ij'}$. And the size of  $c_{ij'}$ grows exponentially with $n$ for a given plaintext, because $C=PM_n$ and the entries of $M_n$ grow exponentially. 
\begin{definition}\label{ChkRng}
The checking range for $c_{ij}$ (relative to $c_{ij'}$) is the interval for its values specified by the checking relations \eqref{minMmaxMn}. 
\end{definition}
To ensure the shrinking of the checking range we need a stronger condition on $L$ than the strong Perron-Frobenius property. It turns out that the range's asymptotic behavior depends on the subdominant eigenvalues of $L$, and it will shrink to a point when they are all located inside the unit disk. This is the Pisot condition \cite{AK08}.
\begin{definition}\label{Pisot}
A polynomial is called Pisot if it has real coefficients, the free term $\geq1$, a simple positive dominant root, and the rest of the roots have absolute values $<1$. We call a linear recurrence or a matrix Pisot if their characteristic polynomial is Pisot.
\end{definition}
\noindent Pisot polynomials with integer coefficients are well studied in number theory and there are explicit formulas and effective algorithms for generating them \cite{Boyd96,HS21}. Second order recurrences with real coefficients, simple positive dominant root $\tau$, and the free term $\pm1$, like the Fibonacci recurrence or unimodular recurrences of \cite{KS17}, are automatically Pisot because the second root is $\pm\frac1\tau$. One can check directly that Tribonacci, $2$-Fibonacci and order $3$ Wielandt recurrences are also Pisot. However, order $4$ Wielandt recurrence is not Pisot, its second largest eigenvalue pair has absolute values $\approx1.06$.
\begin{theorem}\label{RangeShrink} In conditions of Theorem \ref{kDCheckRel} suppose additionally that $L$ is Pisot. Then the checking ranges shrink to a point when $n\to\infty$.
\end{theorem}
\begin{proof} Let $\sigma$ denote the second largest eigenmodulus of $L$. Following the notation of Lemma \ref{RatLim} and keeping only the leading terms of the asymptotic, we compute
\begin{align}\label{RatDiff}
\frac{X_{n+s}^{(i)}}{X_n^{(i)}}-\frac{X_{n+s}^{(j)}}{X_n^{(j)}}
=\frac{(L^{n+s}x_0)^i}{(L^{n}x_0)^i}-\frac{(L^{n+s}x_0)^j}{(L^{n}x_0)^j}
=\frac{x_0^1\sum\limits_{|\lambda_l|=\sigma}Q_{i,j}(n)\left(\frac{\lambda_l}{\tau}\right)^n+\text{LOT}}{(x_0^1)^2\,u^iu^j+\text{LOT}}, 
\end{align}
where  $Q_{i,j}(n)$ are polynomials of degree at most $k-2$, and LOT are exponentially smaller terms. Note that in conditions of Theorem \ref{kDCheckRel} $(x_0^1)^2\,u^iu^j\neq0$. Therefore,
$$
\left|\frac{X_{n+s}^{(i)}}{X_n^{(i)}}-\frac{X_{n+s}^{(j)}}{X_n^{(j)}}\right|
\leq Q(n)\,\left(\frac{\sigma}{\tau}\right)^n,
$$
where the polynomial $Q(n)$ absorbs all the constants. In particular, this estimate holds for the difference between the bounds in \eqref{minMmaxMn}. Since $c_{ij'}$ is a positive linear combination of $X_m^{(i)}$ with $n\leq m\leq n+k-1$ by \eqref{GenMn}, it is bounded by $K\tau^n$ for some $K>0$. Hence the length of the checking range is bounded by $KQ(n)\,\sigma^n$. Under the Pisot condition, $\sigma<1$ and it goes to $0$ when $n\to\infty$.
\end{proof}

\section{Transition ratio and the size of entries}\label{TransRat} 

In view of the checking relations and by analogy to the golden ratio, it is convenient to give the dominant eigenvalue of the left transition matrix a name.
\begin{definition}\label{TransfRat}
Let $L$ be the left transition matrix of a coding sequence $M_n$ with a simple positive dominant eigenvalue $\tau$. We call $\tau$ the transition ratio of $M_n$.
\end{definition}
\noindent The golden ratio $\varphi$ is exactly the dominant eigenvalue of the $Q$-matrix and the transition ratio of the golden matrices $M_n=Q^{n+1}$. In general, when $M_0$ is invertible $L$ and $R$ are intertwined by it and hence similar, so the transition ratio can also be characterized as the dominant eigenvalue of $R$. In particular, like $R$, the transition ratio is determined by the recurrence \eqref{kLinRec} alone. 

Of course, $\tau$ is nothing other than the spectral radius of $L$, which is defined for any matrix. But its relation to the size of entries of $M_n$ is stronger than for arbitrary matrices where the growth of some entries may have no effect on the spectral radius, as in triangular matrices. Since the rows of $M_n$ are segments of recurrent sequences with the limit ratio $\tau$ they asymptotically grow as $K\tau^n$ for some constant $K$. This means that practical considerations require $\tau$ not to be too large if large $n$ are to still be feasible for encryption, see Example \ref{L1234}. Note that for invertible integer valued matrices we always have $\tau\geq1$ because the determinant is at least $1$ by absolute value. In examples appearing in the literature one typically has $\tau<2$. 

We can also give a more straightforward characterization of the transition ratio analogous to the usual definition of the golden ratio, which is convenient for estimating its numerical values.
\begin{definition}\label{StandRec}
Given a linear recurrence \eqref{kLinRec} define its standard sequence $S_n$ as the one satisfying it with the initial values $S_{k-1}=1$, $S_{k-2}=\dots=S_0=0$.
\end{definition}
\noindent For the Fibonacci recurrence the standard sequence will be exactly the sequence of Fibonacci numbers. The authors of \cite{BD14,BP09} use the standard sequences of the Tribonacci and $2$-Fibonacci recurrences, respectively.
\begin{corollary}\label{StandRat} Suppose the left companion matrix \eqref{Ltransf} of a linear $k$-th order recurrence is invertible and has a simple dominant eigenvalue $\tau$, whose eigenvector has strictly positive entries. Let $S_n$ be its standard sequence, then 
\begin{equation}\label{LimStand}
\tau=\lim_{n\to\infty}\frac{S_{n+1}}{S_n},
\end{equation}
and the convergence is exponentially fast. 
\end{corollary}
\begin{proof}We apply Lemma \ref{RatLim} with $M_0$ of the form \eqref{M0Cyc}. For the standard sequence, $x_0=\begin{pmatrix}1 & 0 & \dots & 0\end{pmatrix}^T$ and is cyclic for any invertible companion matrix $L$. Hence $M_0$ and all $M_n=L^nM_0$ are invertible, so all conditions of the lemma are met. In this case, $M_n$ are of the form \eqref{SymMn} with $S_m$ in place of $X_m$. Thus, $\frac{S_{n+1}}{S_n}=\tau+o(n)$ with exponentially small $o(n)$. 
\end{proof}
\noindent This generalizes the property of the Fibonacci numbers and the golden ratio observed already by Kepler. The first result of this sort for general linear recurrences is due to Poincare \cite{FV11}. The same property holds for any recurrent sequence with a cyclic initial vector, not just the standard one. However, without some non-degeneracy conditions the limit ratio may not exist, and may not coincide with the transition ratio even when it does exist \cite{FV11}. This happens because the initial vector $x_0$ may have zero projection on the dominant eigenvector, and the asymptotic behavior is then determined by subdominant eigenvalues and (generalized) eigenvectors. This is ruled out for the standard sequence, and for the row sequences from Lemma \ref{RatLim} generally, because $x_0$ is cyclic, and hence has a non-zero projection on every eigenvector.

\begin{example}\label{TribNonSt}
Using the Tribonacci recurrence $X_{n+3} = X_{n+2} + X_{n+1}+ X_n$ we have the standard sequence using the initial values $S_2 = 1$, $S_1 = 0$, and $S_0 = 0$. We can estimate the transition ratio by taking ratios of consecutive terms in the standard sequence. For example, $n = 20$ gives $\tau\approx\frac{S_{n+1}}{S_n} = \frac{66012}{35890} \approx 1.8393$. This can be used as the initial guess for a non-linear numerical solver to get a more precise value. Solving $z^3=z^2+z+1$ in Maple we find $\tau\approx 1.839286755...$
\end{example}

\section{Error detection and correction}\label{ErrCor}

In this section, we illustrate error detection and correction algorithms for cryptography with higher order recurrences. By the checking relations, for large $n$ all consecutive entries in a row of the ciphertext matrix have ratios close to $\tau$. Hence, if we know even a single correct entry in a row we should be able, in principle, to recover the entire row. Subject to two caveats. First, as we saw already in golden cryptography, the bounds in the checking relations may not be tight enough to pick a single value. Either one has to test all candidate values by trial and error, or transmit additional check data (traditionally, the determinant) to determine them.

Second, to work from a single entry in a row we must know {\it which} entry is correctly transmitted. If we do not know that by other means we need at least two entries in each row to be correct. Then we can detect the correct pair from the fact that their ratio is close to the appropriate power of the transition ratio (determined by how far they are separated in the row). Once this is done, the remaining row entries can be recovered as above. In the second order case, having two correct entries in each row meant no errors at all, and the determinant had to be used to make up for the missing information when they did occur. In the higher order case there is enough built-in redundancy to make transmission of additional check data less appealing. The examples below illustrate some typical situations, when no extra data is transmitted.
\begin{example}[{\bf Single error, unknown location}]\label{3x3Single} Consider the ciphertext matrix from Example \ref{2FibEncryp} transmitted with a single error in the top row. The received matrix is $C^{*} = \begin{pmatrix}
60861 & 41528 & 28373\\
68585 & 46798 & 31933\\
68601 & 46809 & 31940 
\end{pmatrix}$. The transition ratio can be estimated as in Example \ref{TribNonSt} and is $\tau \approx 1.465571232$. Calculation shows that $\frac{c_{12}^*}{c_{13}^*} = \frac{41528}{28373}$ is the only consecutive ratio that is not approximately $\tau$. Since $\frac{c_{11}^*}{c_{12}^*}$ is correct the error has to be in $c_{13}^*$. Taking the ratios of the second and third column entries of $M_{15}=\begin{pmatrix} 406 & 277 & 189\\277 & 189 & 129\\189 & 129 & 88 \end{pmatrix}$ we have by the checking relations \eqref{minMmaxMn}:
$$
1.465\approx\frac{189}{129}\leq\frac{c_{12}}{c_{13}}\leq\frac{129}{88}\approx1.466,
$$
which gives 
$$
28329\approx41528\cdot\frac{88}{129}\leq c_{13}\leq41528\cdot\frac{129}{189}\approx28345.
$$
The checking range has $16$ possible values for $c_{13}$ that can be tried and the one producing a coherent text upon decryption is selected. However, the estimate obtained using the transition ratio, i.e. $41528\cdot\tau\approx28335.7$ is much closer to the correct value $28337$ than the checking bounds, and this is typical. Therefore, a spiral search algorithm centered around this estimate will be faster, and will find the correct value after checking just three candidates. Moreover, the $2$-Fibonacci sequence used here is Pisot, so the checking range shrinks to a point with $n\to\infty$. For $n\geq29$ it contains a single integer. The correct value of $c_{13}$ (which becomes $6,736,252$ when $n=29$ is used for encryption) can then be recovered uniquely.
\end{example}
\noindent There is a tradeoff between having smaller ciphertexts (small $n$) and being able to correct without trial and error (large $n$). Larger $n$ means greater redundancy in the ciphertext, so one expects to have an easier time with correction. However, as we know from Section \ref{PisotCon}, extra redundancy can be reliably leveraged only when $L$ is Pisot.
\begin{example}[{\bf Checking ranges}] Let $\sigma$ denote the second largest eigenmodulus of $L$ used in the proof of Theorem \ref{RangeShrink}. Tribonacci, $2$-Fibonacci and order $3$ Wielandt recurrences have $\sigma\approx0.74,\,0.83$ and $0.87$, respectively. The smallest $n$ for which the checking ranges are $<1$ in length are $19$, $29$, and $43$, respectively. As expected, the smallest $n$ increases with $\sigma$. For many other Pisot recurrences the value of $\sigma$ is even closer to $1$, so the smallest such $n$ will be larger. For the Tribonacci recurrence the dependence of the checking range $[a,b]$ on $n$ is shown in Table \ref{TribCheckRange}. As one can see, its length is not quite monotone decreasing with $n$.
\begin{table}[!ht]
\centering
\begin{tabular}{|c|c|c|c|c|}
\hline
 $n$   & $12$ & $16$ & $17$ & $19$\\  \hline
$a$   & $135,256.10$ &  $1,547,924.82$ & $2,847,074.91$ & $9,631,589.62$ \\ \hline
$b$  & $135,245.10$ & $1,547,923.40$ & $2,847,077.51$ & $9,631,588.74$ \\ \hline
$b-a$  & $11.00$ & $1.42$ & $2.61$ & $0.87$ \\ \hline
\end{tabular}
\caption{\label{TribCheckRange} Checking ranges for the Tribonacci recurrence at various $n$.}
\end{table}
\end{example}
\begin{example}[{\bf Double error, unknown location}] Suppose the word EXTRATERRESTRIAL is encrypted using symmetric coding matrices of the Tetranacci recurrence $X_{n+4} = X_{n+3}+X_{n+2}+X_{n+1}+X_n$ with the standard initial vector and $n=5$. Two entries in the top row of the ciphertext matrix are transmitted with errors, and the transmitted matrix is 
$$
C^{*} = \begin{pmatrix}
16460 & 8332 & 4123 & 2239\\
14955 & 7767 & 4025 & 2087\\
16387 & 8510 & 4413 & 2282\\
15969 & 8292 & 4297 & 2226\\
\end{pmatrix}\!.
$$
The transition ratio of the Tetranacci sequence is $\tau\approx1.927562$, and the second eigenmodulus is $\sigma\approx0.8182726$. The consecutive ratios in the top row are $\frac{c_{11}}{c_{12}}\approx1.9755$, $\frac{c_{12}}{c_{13}}\approx2.0208$, $\frac{c_{13}}{c_{14}}\approx1.8439$, and they are all incorrect. The square of the transition ratio is $\tau^2\approx3.715495$, and $\frac{c_{11}}{c_{13}}\approx3.9922$, $\frac{c_{12}}{c_{14}}\approx3.7262$. The last ratio is close to the limit value (correct and incorrect ratios are more clearly separated for larger $n$), so we conclude that the errors are in $c_{11}$ and $c_{13}$ entries.

To recover the correct values we use the same approach as in Example \ref{3x3Single}. The coding matrix is
$$
M_5 = \begin{pmatrix}
108 & 56 & 29 & 15\\
56 & 29 & 15 & 8\\
29 & 15 & 8 & 4\\
15 & 8 & 4 & 2\\
\end{pmatrix}\!.
$$
For $c_{11}$ the reference value is $c_{12}=8332$, and the corresponding ranges are 
$$
1.875\approx\frac{15}{8}\leq\frac{c_{11}}{c_{12}}\leq\frac{29}{15}\approx1.9333,
$$
and 
$$
15623\approx8332\cdot\frac{15}{8}\leq c_{11}\leq8332\cdot\frac{29}{15}\approx16108.
$$
The case of $c_{13}$ is analogous with the reference value $c_{14}=2239$. The starting value for the spiral search should again be $c_{12}\cdot\tau\approx 16060.4$, which is closer to the correct value $16046$ than the checking bounds. About thirty candidate values would have to be tested, and, in this case, it has to be done in combination with testing candidates for $c_{13}$. Picking a larger $n$ would narrow down the checking ranges and make the correction easier. 
\end{example}
The correction algorithm is somewhat modified when the error location is known, and, as already mentioned, more errors in the same row can be corrected. In fact, one just needs to have a single correct entry in each row. The next example also illustrates that, and the use of non-consecutive ratios for error correction.
\begin{example}[{\bf Triple error, known location}] As in the previous example, we encrypted EXTRATERRESTRIAL with the symmetric coding matrices of the Tetranacci sequence and $n=5$. The transmitted ciphertext matrix is
$$
C^{*} = \begin{pmatrix}
16046 & 4513 & 7211 & 1337\\
14955 & 7767 & 4025 & 2087\\
16387 & 8510 & 4413 & 2282\\
15969 & 8292 & 4297 & 2226
\end{pmatrix}
$$
There are three errors in the top row, but $c_{11}$ is known to be transmitted correctly. The rough estimates for the remaining entries are $c_{12}\approx\frac{c_{11}}{\tau}\approx8326$, $c_{13}\approx\frac{c_{11}}{\tau^2}\approx4321$, and $c_{14}\approx\frac{c_{11}}{\tau^3}\approx2239$. 

The checking ranges are calculated by taking maximal and minimal ratios in the corresponding columns of $M_5$. 
\begin{align*}
8299.66\approx16046\cdot\frac{29}{15}\leq c_{12}\leq16046\cdot\frac{15}{8}\approx8557.87\\
4278.93\approx16046\cdot\frac{15}{4}\leq c_{13}\leq16046\cdot\frac{29}{8}\approx4426.48 \\
2139.47\approx16046\cdot\frac{15}{2}\leq c_{14}\leq16046\cdot\frac{56}{8}\approx2292.29
\end{align*}
The ranges are wide because $n$ is small, but they shrink to containing a single integer for $n\geq34$.
\end{example}
Additional check data can also be transmitted to avoid or narrow down trial and error searches. In particular, one can use  determinants to correct single errors by solving a linear equation, or double errors by solving Diophantine equations within the range of solutions narrowed down by the checking relations as in \cite{KS17}. However, this requires one to assume that all other entries, and the determinant itself, are correctly transmitted. Considering also the inefficiency of calculations with determinants in higher dimensions, it seems more reasonable to transmit row sums of $C$ instead as additional check data. For example, this would allow correcting a single error by solving a linear equation  {\it without} assuming that all other rows are correctly transmitted.

\section{Generation of the coding matrices}\label{GenEnc}

In this section, we review some approaches to generating coding matrices that meet the conditions of Theorems \ref{kDCheckRel} and \ref{RangeShrink}, and hence are suitable for encryption with error detection and correction.

\subsection{Companion matrices}

Generation is quite straightforward if one is content with using symmetric coding matrices. When $L$ is a left companion matrix \eqref{Ltransf} then, by direct computation, its eigenvectors are of the form $\begin{pmatrix}\l^{k-1} & \dots & \l & 1\end{pmatrix}^T$, where $\l$ is the eigenvalue. Therefore, if the characteristic polynomial of $L$ has a simple positive dominant root $\tau$ the dominant eigenvector will automatically have strictly positive entries, i.e. $L$ will be strong Perron-Frobenius. Generation of such polynomials with integer coefficients is discussed in \cite{DS15}. Writing the polynomial as $f(z):=z^k-a_{k-1}z^{k-1}-\dots-a_1z-a_0$, one can simply pick $a_i$ randomly from a range $[-M,N]$ and select those that have small positive dominant roots, say $\tau<3$. The Pisot condition can be used to further restrict selection.

Invertibility is ensured by restricting to $a_0\neq0$. Moreover, left companion matrices always have a cyclic vector, namely $\begin{pmatrix}1 & 0& \dots & 0\end{pmatrix}^T$, and, in fact, almost every vector is cyclic. Finding cyclic vectors deterministically is discussed in \cite{AC97}.
\begin{example}\label{Comp1121} For this example we generated polynomials with $a_i$ selected from $[0,3]$. Skipping those with $0$-$1$ coefficients we took $$f(z) = z^4 -z^3 -z^2 -2z -1.$$ Its dominant root is $\tau\approx2.066$ and the second largest eigenmodulus is $\sigma\approx0.9582$, i.e. this polynomial is Pisot. Taking the standard initial vector $x_0 = \begin{pmatrix}
1 & 0 & 0 & 0 \end{pmatrix}^T$ we generate 
$$L= \begin{pmatrix}
1 & 1 & 2 & 1\\
1 & 0 & 0 & 0\\
0 & 1 & 0 & 0\\
0 & 0 & 1 & 0
\end{pmatrix}\ \text{and}\ M_0 = \begin{pmatrix}
5 & 2 & 1 & 1\\
2 & 1 & 1 & 0\\
1 & 1 & 0 & 0\\
1 & 0 & 0 & 0
\end{pmatrix}\!.
$$ 
The last component of the key, positive integer $n$, can also be generated randomly from a large range $[M,N]$. When it is selected the rows of $M_0$ provide initial values to generate  the sequences $X_m$ recursively up to $m=n+6$. In this case, we calculate for $n=10$\,:
\begin{align*}
&X_m=0,0,0,1,1,2,5,10,20,42,87,179,370,765,1580,3264,6744,\dots;\\
&M_{10} = \begin{pmatrix}
6744 & 3264 & 1580 & 765\\
3264 & 1580 & 765 & 370\\
1580 & 765 & 370 & 179\\
765 & 370 & 179 & 87
\end{pmatrix}\!.
\end{align*}
\end{example}

\subsection{Pisot polynomials}

If one wishes to use only Pisot transition matrices there are more selective methods to generate them without sieving through general polynomials. For example, there is a complete description of Pisot polynomials with integer coefficients and $1<\tau<2$ due to Amara, Boyd and Talmoudi \cite{Boyd96,HS21}. The set of limit points of their dominant roots, called Pisot numbers, consists of two infinite series $\phi_r$ and $\psi_r$ and one exceptional value $\chi$. The Pisot polynomials for $\psi_r$ are none other than the characteristic polynomials  $\Psi_r(z)=z^{r+1}-z^{r}-\dots-z-1$ of the ($r+1$)-bonacci recurrences. Pisot numbers in a neighborhood of $\psi_r$ are roots of $z^m\,\Psi_r(z)\pm(z^{r+1}-1)$ or $z^m\,\Psi_r(z)\pm\frac{z^r-1}{z-1}$ (for small $m$ we may have $\tau>2$), with these polynomials either being themselves Pisot or becoming Pisot after dividing out some trivial factors. There are similar descriptions for Pisot numbers associated with $\phi_r$ and $\chi$. The (finitely many) irregular Pisot numbers not covered are also known.
\begin{example} Consider polynomials  $f_m(z)=z^m\,\Psi_2(z)-(z^3-1)$ associated with the Tribonacci recurrence. They are Pisot if $m$ is odd, and $\frac{f_m(z)}{z+1}$ are if $m$ is even \cite{Boyd96}. We skip $m=1,2$, where $\tau>2$, and for $m=3,4$ find:
\begin{align*}
f_3(z)&=z^6-z^5-z^4-2z^{3}+1;\ \ \ \tau\approx1.98139,\sigma\approx0.94792\\
\frac{f_4(z)}{z+1}&=z^6-2z^5+z^4-2z^{3}+z^2-z+1;\ \ \ \tau\approx1.91616,\sigma\approx0.93460
\end{align*}
Note that some $a_i<0$, including the free term. Forming the left companion matrices and selecting some cyclic vector $x_0$ one can calculate the remaining key data as in Example \ref{Comp1121}.
\end{example}

\subsection{Primitive matrices}

Recall that a matrix with non-negative entries is called {\it primitive} when some power of it has strictly positive entries. By a theorem of Perron, all primitive matrices have the strong Perron-Frobenius property \cite[9.4]{Lanc}. One can see that the primitivity only depends on where $0$ entries are located in the matrix. Hence, to test for the primitivity it is sufficient to inspect $0$-$1$ matrices. Those, in turn, can be associated to directed graphs (with self-loops) as their adjacency matrices.
\begin{definition}\label{StrCon}
A directed graph is called strongly connected if there is a directed path in it from any vertex to any other vertex, and it is called acyclic if the greatest common divisor of the lengths of directed cycles in it is $1$.
\end{definition}
\noindent A $0$-$1$ matrix is primitive if and only if its directed graph is strongly connected and acyclic \cite[2.4]{Varg}. Thus, the primitivity of a matrix can be checked quite easily. In particular, a $k\times k$ left companion matrix $L$ with non-negative coefficients $a_i\geq0$ is primitive if and only if $a_0 > 0$ and $a_{i}>0$ for some $i>0$ relatively prime to $k$. Directed graphs of some left companion matrices are shown on Figure \ref{3x3Graphs}. All of them are strongly connected, and all but c) are acyclic.
\begin{figure}[!ht]
\begin{centering}
a)\includegraphics[width=2.8cm]{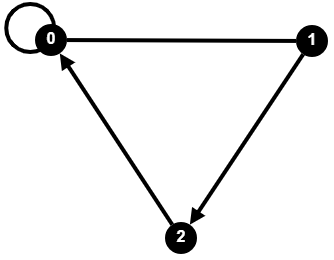}  b)\includegraphics[width=2.8cm]{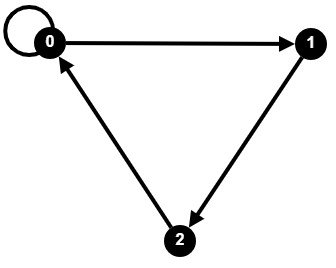}
c)\includegraphics[width=2.8cm]{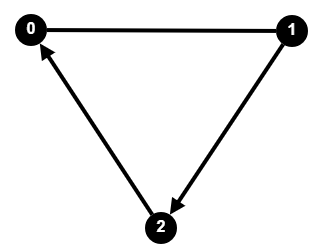}
d)\includegraphics[width=2.4cm]{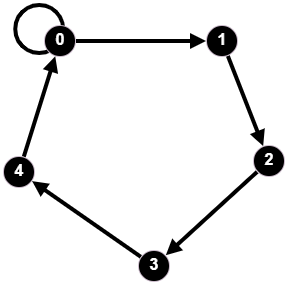}
\par\end{centering}
\hspace*{-0.1in}\caption{\label{3x3Graphs} Directed graphs of left companion matrices of a) Tribonacci; b) $2$-Fibonacci; c) recurrence with  $a_0=a_1 = 1$, $a_2 = 0$; d) $4$-Fibonacci.}
\end{figure}

Primitive $0$-$1$ matrices can be used as seeds to generate primitive matrices with any non-negative entries because replacing entries by larger entries preserves the primitivity. It also increases $\tau$, which can serve as a limiting factor in a generation algorithm. Since $\tau$ is bounded above by the matrix's maximal row sum and by its maximal column sum \cite[2.1]{Minc}, generation should favor sparse matrices with many zeros and low non-zero entries. Imposing the Pisot condition narrows the selection further. 
\begin{example}\label{WielSeedL} Starting with the order $3$ Wielandt left transition matrix $\begin{pmatrix}
0  & 1 & 1\\
1  & 0 & 0\\
0  & 1 & 0
\end{pmatrix}$ we increase some of the $0$ entries in it, and check for the size of $\tau$. One of the non-companion matrices so generated is $L=\begin{pmatrix}
0  & 1 & 1\\
1  & 2 & 1\\
0  & 1 & 0
\end{pmatrix}$. Although $\tau\approx2.831>2$, it is not too large, and in return this matrix is Pisot with a fairly small second eigenmodulus $\sigma\approx0.594$. The initial vector $\begin{pmatrix}0 & 0& 1\end{pmatrix}^T$ is cyclic for $L$, and $M_0=\begin{pmatrix}
1  & 1 & 0\\
3  & 1 & 0\\
1  & 0 & 1
\end{pmatrix}$. The characteristic polynomial of $L$ is 
$f(z) = z^3-2z^2-2z-1$, which means that the recurrence is $X_{n+3}=2X_{n+2}+2X_{n+1}+X_n$.
Since $L$ is not a left companion matrix the rows of $M_n$ are formed by three (potentially) different sequences satisfying this recurrence with the initial values given by the rows of $M_0$, see \eqref{GenMn}. Generating them recursively up to the $6$-th term we compute:
\begin{align*}
&X_m^{(2)}=0,1,1,4,11,31,88,249,705,\dots\\ 
&X_m^{(1)}=0,1,3,8,23,65,184,521,1475,\dots\\
&X_m^{(0)}=1,0,1,3,8,23,65,184,521,\dots;\\
&M_{6}=\begin{pmatrix}
705 & 249 & 88\\
1475 & 521 & 184\\
521 & 184 & 65
\end{pmatrix}\!.
\end{align*}
\end{example}

\subsection{Right companion representation}

The methods discussed above generated $M_n$ from left transition matrices. However, one can use the alternative representation described in Section \ref{RiteComp} instead, and start from a right companion matrix $R$, and an initial matrix $M_0$. The benefit is that $M_0$ does not have to be computed from a cyclic vector and can be generated randomly. The approach is based on the following lemma. 
\begin{lemma}\label{Rgen}
Let $R$ be an invertible right companion matrix of the form \eqref{Rtransf} with the strong Perron-Frobenius property. Let $M_0$ be any invertible matrix with non-negative entries and set $M_n:=M_0R^n$. Then the left transition matrix of $M_n$ also has the strong Perron-Frobenius property.
\end{lemma}
\begin{proof} By construction, $M_n$ are invertible. Set $L:=M_0RM_0^{-1}$, then $L$ is similar to $R$ and has the same eigenvalues with the same multiplicities. In particular, it has the same simple dominant eigenvalue $\tau$ as $R$. We just have to show that the corresponding eigenvector has strictly positive entries. Let $\xi$ be the dominant eigenvector of $R$. Then
$LM_0\xi=M_0R\xi=\tau M_0\xi$, so $M_0\xi$ is the dominant eigenvector of $L$. Since $\xi$ has strictly positive entries and $M_0$ has non-negative entries by assumption, $M_0\xi$ also has non-negative entries. Moreover, they can be $0$ only if $M_0$ has a row of all $0$-s, which contradicts its invertibility. 
\end{proof}
In the following example we use the transpose of a left companion matrix as both $R$ and $M_0$.
\begin{example}\label{R102}Take $R=\begin{pmatrix}
1  & 1 & 0\\
0  & 0 & 1\\
2  & 0 & 0
\end{pmatrix}$ and $M_0:=R$. It follows that $L=R$ and $M_n:=R^{n+1}$. Note that $L$ is not a left companion matrix. The transition ratio is $\tau \approx 1.6956$, and the corresponding eigenvector (computed with Maple) is approximately $\begin{pmatrix}0.84781&0.58975&1\end{pmatrix}^T$. The coding matrices can be generated as in Example \ref{WielSeedL}.
\end{example}

\section{Conclusions}

We introduced a generalization of golden cryptography to higher order linear recurrences, and gave explicit conditions for preserving its error correction properties. We also described a number of ways to randomly generate coding matrices that satisfy these conditions, and algorithms for error detection and correction. 

In the simpler case the coding matrices $M_n$ are symmetric and filled with consecutive entries of a single recurrent sequence. The encryption key consists, in addition to the index $n$, of $k$ coefficients of the recurrence and $k$ initial values, where $k$ is the order of the recurrence. In the most general case, the rows of $M_n$ are segments of different sequences, albeit satisfying the same recurrence, and $k$ initial values are replaced by $k^2$ ones, $k$ for each row sequence. Additional parameters greatly improve the security of encryption, which is no longer vulnerable to the known types of brute force and chosen plaintext attacks. The tradeoff is the increased computational burden, especially when $k$ and $n$ are large. 

Although our coding matrices can be represented in the form $M_n=L^nM_0$, their special recurrent structure ensures that matrix multiplication is not needed to compute them. This makes the encryption/decryption procedures computationally attractive, especially for large $k$ and $n$. At the same time, our main results show that  certain spectral properties of $L$ play a central role in determining which coding matrix sequences are feasible in practice. First, there is a tradeoff between the size of the spectral radius $\tau$ of $L$ and the size of indices $n$ that can be used without producing intractably large ciphertexts. Second, $L$ must have the strong Perron-Frobenius property to induce the checking relations in the ciphertext leveraged by the error detection and correction algorithms. Finally, an even stronger Pisot condition on $L$ is needed to minimize trial and error in those algorithms, and utilize encryption redundancy most efficiently. 

These restrictions, especially the last one, somewhat reduce the variety of feasible coding matrices compared to merely formal considerations. In particular, if we bound the sizes of both $\tau$ and $k$ then there are only finitely many integer matrices $L$ that are strong Perron-Frobenius or Pisot. However, their number grows quickly with $k$, and the additional freedom of choosing the initial values still provides abundant means for ensuring encryption security.

Finally, we presented a number of methods for generating feasible coding matrices. Most of them favor cases where $L$ has non-negative entries. However, the companion matrices of Pisot polynomials often have negative entries, and finding systematic ways of generating strong Perron-Frobenius and Pisot matrices of more general form with some negative entries is desirable. Especially those that allow controlling the sizes of the spectral radius and of the second largest eigenmodulus.

One may also wish to explore coding matrices of the form $M_n=L^nM_0$, where no relation between $L$ and $M_0$ is assumed (in our setting $M_0$ is always generated by a single cyclic vector of $L$). In that case, the entries of $M_n$ may no longer be terms of recurrent sequences, and hence may be harder to compute. It is also unclear how one can establish the checking relations in this generality. However, in some special cases, at least, a version of them may still hold. For example, in \cite{BD14} the right companion matrix of the Tribonacci sequence is used as both $L$ and $M_0$, and the row ratios of the ciphertext matrices are shown to approach some rational functions of $\tau$. It is of interest to characterize classes of pairs $L$, $M_0$ that admit such generalizations of the checking relations.

\end{document}